\documentclass [11pt]{article}
\usepackage[utf8]{inputenc}
\usepackage{fullpage}
\usepackage{amsmath}
\usepackage{cite}

\usepackage{amssymb}
\usepackage{graphicx}
\usepackage{mathabx}

\newtheorem{lemma}{Lemma}
\newtheorem{theorem}{Theorem}
\newtheorem{corollary}{Corollary}

\newcommand{\qed}{\hfill\ensuremath{\Box}\medskip\\\noindent}
\newenvironment{proof}{\noindent\emph{Proof. }}{}

\newcommand{\nca}{\ensuremath{\mathrm{nca}}}

\newcommand{\Taux}{\ensuremath{\widetilde{T}}}

\newcommand{\D}{\ensuremath{D}}
\newcommand{\TT}{\ensuremath{\mathcal{T}}}
\newcommand{\TD}{\ensuremath{\mathcal{TD}}}
\newcommand{\nt}{\ensuremath{n_T}}
\newcommand{\nd}{\ensuremath{n_{\D}}}
\newcommand{\ntd}{\ensuremath{n_\TD}}
\newcommand{\ntt}{\ensuremath{n_\TT}}

\newcommand{\Access}{\ensuremath\mathsf{Access}}
\newcommand{\Decompress}{\ensuremath\mathsf{Decompress}}
\newcommand{\Parent}{\ensuremath\mathsf{Parent}}
\newcommand{\Depth}{\ensuremath\mathsf{Depth}}
\newcommand{\Height}{\ensuremath\mathsf{Height}}
\newcommand{\Size}{\ensuremath\mathsf{Size}}
\newcommand{\FirstChild}{\ensuremath\mathsf{Firstchild}}
\newcommand{\NextSibling}{\ensuremath\mathsf{NextSibling}}
\newcommand{\LevelAncestor}{\ensuremath\mathsf{LevelAncestor}}
\newcommand{\NCA}{\ensuremath\mathsf{NCA}}
\newcommand{\decomp}{\ensuremath\mathsf{Decompress}}
\newcommand{\FindRepr}{\ensuremath\mathsf{FindRepresentatives}}

\usepackage[dvipsnames,usenames]{color}
\usepackage[colorlinks=true,urlcolor=Blue,citecolor=Green,linkcolor=BrickRed]{hyperref}
\begin{document}
\title{Tree Compression with Top Trees\thanks{An extended abstract of this paper appeared at the 40th International Colloquium on Automata, Languages and Programming.}}
\author{Philip Bille \thanks{Partially supported by the Danish Agency for Science, Technology and Innovation.} \\  \href{mailto:phbi@dtu.dk}{phbi@dtu.dk} \and Inge Li G{\o}rtz$^{\dag}$\\
 \href{mailto:inge@dtu.dk}{inge@dtu.dk}     \and Gad M. Landau\thanks{Partially supported by the National Science Foundation
Award 0904246, Israel Science Foundation grant 347/09,
Yahoo, Grant No. 2008217 from the United States-Israel
Binational Science Foundation (BSF) and DFG.} \\  \href{mailto:landau@cs.haifa.ac.il}{landau@cs.haifa.ac.il}   \and Oren Weimann $^{\dag}$ \thanks{Partially supported by  the Israel Science Foundation grant 794/13.} \\  \href{mailto:oren@cs.haifa.ac.il}{oren@cs.haifa.ac.il} }
\date{}

\maketitle

\begin{abstract}

We introduce a new compression scheme for labeled trees based on top trees~\cite{TopTrees}. 
Our  compression scheme is the first to simultaneously take advantage of internal repeats in the tree (as opposed to the classical DAG compression that only exploits rooted subtree repeats) while also supporting fast navigational queries directly on the compressed representation. We show that the new compression scheme achieves close to optimal worst-case compression, can compress exponentially better than DAG compression, is never much worse than DAG compression, and supports navigational queries in logarithmic time.
\end{abstract}

\section{Introduction} 
A labeled tree $T$ is a rooted, ordered tree, where each node has a label from an alphabet $\Sigma$. Labeled trees appear in computer science as tries, dictionaries, parse trees, suffix trees, XML trees, etc. In this paper, we study compression schemes for labeled trees that take advantage of \emph{repeated substructures} and support navigational queries, such as returning the label of a node $v$, the parent of $v$, the depth of $v$, the size $v$'s subtrees, etc., directly on the compressed representation. We consider the following two basic types of repeated substructures (see Figure~\ref{fig:subtree}). The first type is used in DAG compression~\cite{BKG03,FKG03} and the second in tree grammars~\cite{Busatto04grammarBasedtree,BLM08,LohreyEtAl,LM06,MB04}.

\begin{description}
\item[Subtree repeat.] A \emph{rooted subtree} is a subgraph of $T$ consisting of a node and all its descendants. A \emph{subtree repeat} is an identical (both in structure and in labels) occurrence of a rooted subtree in $T$. 

\item[Tree pattern repeat.] A  \emph{tree pattern} is any connected subgraph of $T$. A  \emph{tree pattern repeat} is an identical (both in structure and in labels) occurrence of a tree pattern in $T$.
\end{description}

In this paper, we introduce a simple new compression scheme, called \emph{top tree compression}, that exploits tree pattern repeats. Compared to the existing techniques our compression scheme has the following advantages: Let $T$ be a tree of size $n$ with nodes labeled from an alphabet of size $\sigma$. We support navigational queries in $O(\log n)$ time (a similar result is not known for tree grammars), the compression ratio is in the worst case at least $\log_\sigma^{0.19} n$ (no such result is known for either DAG compression or tree grammars), our scheme can compress exponentially better than DAG compression, and the compression ratio is never worse than DAG compression by more than a $\log n$ factor.

\subsection{Previous Work}
The previous work on tree compression can be described by three major approaches: using subtree repeats, using tree pattern repeats, and using succinct data structures.  Below we briefly discuss these approaches and the existing  tree compression schemes. 
Extensive practical work has recently been done on all these tree compression schemes (see e.g., the recent survey of Sakr~\cite{sakr2009xml}).

\paragraph{DAG compression.} Using subtree repeats, a node in the tree $T$ that has a child with subtree $T'$ can instead point to any other occurrence of $T'$. This way, it is possible to represent $T$ as a Directed Acyclic Graph (DAG). Over all possible DAGs that can represent $T$, the smallest one is unique and can be computed in $O(n)$ time~\cite{DST1980}. Its size can be exponentially smaller than $n$. DAG representation of trees are broadly used for identifying and sharing common subexpressions, e.g., in programming languages~\cite{steven1997advanced}  and binary decision diagrams~\cite{meinel1998algorithms}. Compression based on DAGs has also been studied more recently in~\cite{BKG03,FKG03, lohrey2013xml} and a Lempel-Ziv analog of subtree repeats was suggested in~\cite{LZGonzalo}. It is possible to support navigational queries~\cite{SODA2011} and path queries~\cite{BKG03} directly on the DAG representation in logarithmic time. 
The problem with subtree repeats is that we can miss many internal repeats. Consider for example the case where $T$ is a single path of $n$ nodes with the same label. Even though $T$ is highly compressible (we can represent it by just  storing the label and the path length) it does not contain a single subtree repeat and its minimal DAG is of size $n$. 
\begin{figure}[t]
   \centering
   \includegraphics[scale=1]{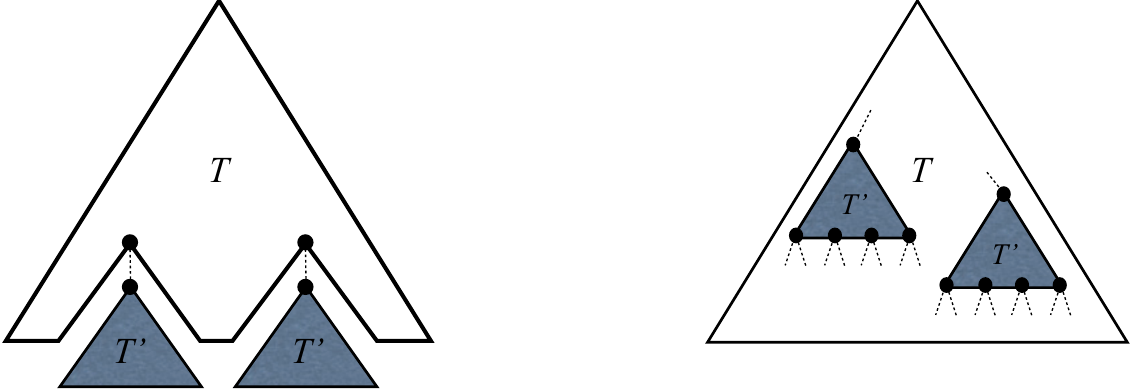}
   \caption{A tree $T$ with a {\em subtree} repeat $T'$ (left), and a  {\em tree pattern} repeat $T'$ (right).}
   \label{fig:subtree}
\end{figure}
\paragraph{Tree grammars.} Alternatively, \emph{tree grammars} are capable of exploiting tree pattern repeats.  Tree grammars generalize grammars from deriving strings to deriving trees and were studied in~\cite{Busatto04grammarBasedtree, BLM08, LohreyEtAl, LM06, MB04}. Compared to DAG compression, a tree grammar can be  exponentially smaller than the minimal DAG~\cite{LM06}. Unfortunately, computing a minimal tree grammar is NP-Hard~\cite{CLL+05}, and all known tree grammar based compression schemes can only support navigational queries in time proportional to the height of the grammar which can be $\Omega(n)$.

\paragraph{Succinct  data structures.} A different approach to tree compression is \emph{succinct data structures} that compactly encode trees. Jacobson~\cite{Jacobson} was the first to observe that the naive pointer-based tree representation using $\Theta(n \log n)$ bits is wasteful. He showed that {\em unlabeled} trees can be represented using $2n+o(n)$ bits and support various queries by inspection of $\Theta(\lg n)$ bits in the bit probe model. This space bound is asymptotically optimal with the information-theoretic lower bound. Munro and Raman~\cite{MunroRaman01} showed how to achieve the same bound in the RAM model while using only constant time for queries. Such representations are called \emph{succinct data structures}, and have  been generalized to include a richer set of queries such as subtree-size queries~\cite{MunroRaman01,BenoitEtAl05} and level-ancestor queries~\cite{Geary04succinctordinal}.
For {\em labeled} trees, Ferragina et al.~\cite{FerraginaFOCS2005} gave a representation using $2n \log \sigma + O(n)$ bits that
supports basic navigational operations, such as find the parent of node $v$, the $i$'th child of $v$, and any child of $v$ with label $\alpha$. Ferragina et al. also introduced the notion of $k$'th order {\em tree entropy} $H_k$ in a restricted model. In this model, used by popular  XML compressors~\cite{XML2,XML1}, the label of a node is a function of the labels of all its ancestors. For such a tree $T$,  Ferragina et al. gave a representation requiring at most $n  H_k(T ) +  2.01n + o(n)$ bits. Note that  the above space bounds do not guarantee a  compact representation when the input contains many subtree repeats or tree pattern repeats. In particular, the total space is never $o(n)$ bits.

\subsection{Our Results.}
We propose a new compression scheme for labeled trees, which we call \emph{top tree compression}. To the best of our knowledge, this is the first compression scheme for trees that (i) takes advantage of tree pattern repeats (like tree grammars) but (ii) simultaneously supports navigational queries on the compressed representation in logarithmic time (like DAG compression). 
In the worst case, we show  that (iii) the compression ratio of top tree compression is always at least $\log_\sigma^{0.19} n$ (compared to the information-theoretic bound of $\log_\sigma n$). This is in contrast to both tree grammars and DAG compression that have not yet been proven to have worst-case compression performance comparable to the information-theoretic bound. Finally, we compare the performance of top tree compression to DAG compression. We show that top tree compression (iv) can compress exponentially better than DAG compression, and (v) is never  worse than DAG compression by more than a  $\log  n$ factor. 
%
%

The key idea in top tree compression is to transform the input tree $T$ into another tree $\TT$ such that tree pattern repeats in $T$ become subtree repeats in $\TT$. The transformation is based on  top trees~\cite{AHLT1997,AHT2000,TopTrees} -- a data structure originally designed for dynamic (uncompressed) trees. After the transformation, we compress the new tree $\TT$ using the classical DAG compression resulting in the \emph{top DAG} $\TD$. The top DAG $\TD$ forms the basis for our compression scheme. We obtain our bounds on compression (iii), (iv), and (v) by analyzing the size of $\TD$ , and we obtain efficient navigational queries (ii) by augmenting $\TD$ with additional data structures.

To state our bounds, let $n_G$ denote the total size (vertices plus edges) of the graph $G$. 
We assume a standard word RAM model of computation with logarithmic word size. All space complexities refer to the number of words used by the data structure. 
We first show the following worst-case compression bound achieved by the top DAG.

\begin{theorem}\label{thm:worstcasebounds}
Let $T$ be any ordered tree with nodes labeled from an alphabet of size $\sigma$ and let $\TD$ be the corresponding top DAG. Then, $n_{\TD} = O(\nt/\log_\sigma^{0.19}\!\nt)$.
\end{theorem}
This worst-case performance of the top DAG should be compared to the information-theoretic lower bound of $\Omega(\nt/\log_\sigma \nt)$. This lower bound applies already for strings (so it clearly holds for labeled trees). It is obtained by simply noticing that there are $\Omega(\sigma^{\nt})$ string of length $\nt$ over an alphabet of size $\sigma$, implying a lower bound of  $\Omega(\nt \log \sigma)$ bits or $\Omega(\nt /\log_\sigma \nt)$ words. Note that with standard DAG compression the worst-case bound is $\Theta(\nt)$ since a single path is incompressible using subtree repeats.

Secondly, we compare top DAG compression to standard DAG compression. 
\begin{theorem}\label{thm:comparison}
Let $T$ be any ordered tree and let $D$ and $\TD$ be the corresponding DAG and top DAG, respectively. For any tree $T$ we have  $\ntd = O(\log \nt) \cdot \nd$ and there exist families of trees $T$ such that $\nd = \Omega (\nt/ \log \nt) \cdot \ntd$.
\end{theorem}
Thus, top DAG compression can be exponentially better than DAG compression (since it's possible that $\nd =O(\log \nt)$) and it is always within a logarithmic factor of DAG compression.  
To the best of our knowledge this is the first non-trivial bound shown for any tree compression scheme compared to the DAG.

Finally, we show how to represent the top DAG $\TD$ in $O(\ntd)$ space such that we can quickly answer a wide range of queries about $T$ without decompressing. 
\begin{theorem}\label{thm:navigation}
Let $T$ be an ordered tree with top DAG $\TD$. There is an $O(n_{\TD})$ space representation of $T$ that supports $\Access$, $\Depth$, $\Height$, $\Size$, $\Parent$, $\FirstChild$, $\NextSibling$, $\LevelAncestor$, and $\NCA$ in $O(\log n_{T})$ time. Furthermore, we can  $\Decompress$ a subtree $T'$ of $T$ in time $O(\log \nt + |T'|)$. 
\end{theorem}
The operations $\Access$, $\Depth$, $\Height$, $\Size$, $\Parent$, $\FirstChild$, and $\NextSibling$ all take a node $v$ in $T$ as input\footnote{The nodes of $T$ are uniquely identified by their preorder numbers. See Section~\ref{section4}.} and return its label, its depth, its height, the size of its subtree, its parent, its first child, and its sibling to the right, respectively. The $\LevelAncestor$ returns an ancestor at a specified distance from $v$, and $\NCA$ returns the nearest common ancestor to a given pair of nodes. Finally, the $\Decompress$ operation decompresses and returns any rooted subtree. 

\section{Top Trees and Top DAGs}

Top trees were introduced by Alstrup et al.~\cite{AHLT1997,AHT2000,TopTrees} for maintaining an uncompressed, unordered, and unlabeled tree under link and cut operations. We extend them to ordered and labeled trees, and then introduce top DAGs for compression. Our construction is related to well-known algorithms for top tree construction, but modified for our purposes. In particular, we need to carefully order the steps of the construction to guarantee efficient compression, and we disallow some combination of cluster merges to ensure fast navigation. 

\subsection{Clusters} 
Let $v$ be a node in $T$ with children $v_1, \ldots, v_k$ in left-to-right order. Define $T(v)$ to be the subtree induced by $v$ and all proper descendants of $v$. Define $F(v)$ to be the forest induced by all proper descendants of $v$. For $1\leq s \leq r \leq k$ let $T(v, v_s, v_r)$ be the  tree pattern induced by the nodes $\{v\} \cup T(v_s) \cup T(v_{s+1}) \cup \cdots \cup T(v_r)$.

A {\em cluster} with {\em top boundary node} $v$ is a tree pattern of the form $T(v, v_s, v_r)$, $1\leq s \leq r \leq k$. A {\em cluster} with {\em top boundary node} $v$ and {\em bottom boundary node} $u$ is a tree pattern of the form $T(v, v_s, v_r) \setminus F(u)$, $1 \leq s \leq r \leq k$, where $u$ is a node in $ T(v_s) \cup  \cdots \cup T(v_r)$.
Clusters can therefore have either one or two boundary nodes. For example, let $p(v)$ denote the parent of $v$ then a single edge 
$(v,p(v))$  of $T$ is a cluster where $p(v)$ is the top boundary node. If $v$ is a leaf then there is no bottom boundary node, otherwise $v$ is a bottom boundary node. Nodes that are not boundary nodes are called 
\emph{internal} nodes. 

Two edge disjoint clusters $A$ and $B$ whose vertices overlap on a single boundary node can be \emph{merged} if their union $C = A \cup B$ is also a cluster. There are five ways of merging clusters, as illustrated by Fig.~\ref{fig:merge}. Merges of type (a) and (b) can be done if the common boundary node is not a boundary node of any other cluster except $A$ and $B$. Merges of type (c),(d), and (e) can be done only if at least one of $A$ or $B$ does not have a bottom boundary node.
The original paper on top trees~\cite{AHLT1997,AHT2000,TopTrees} contains more ways to merge clusters, but allowing these would lead to a violation of our definition of clusters as a tree pattern of the form $T(v, v_s, v_r) \setminus F(u)$, which we need for navigational purposes.
\begin{figure}[tb]
   \centering
   \includegraphics[scale=0.4]{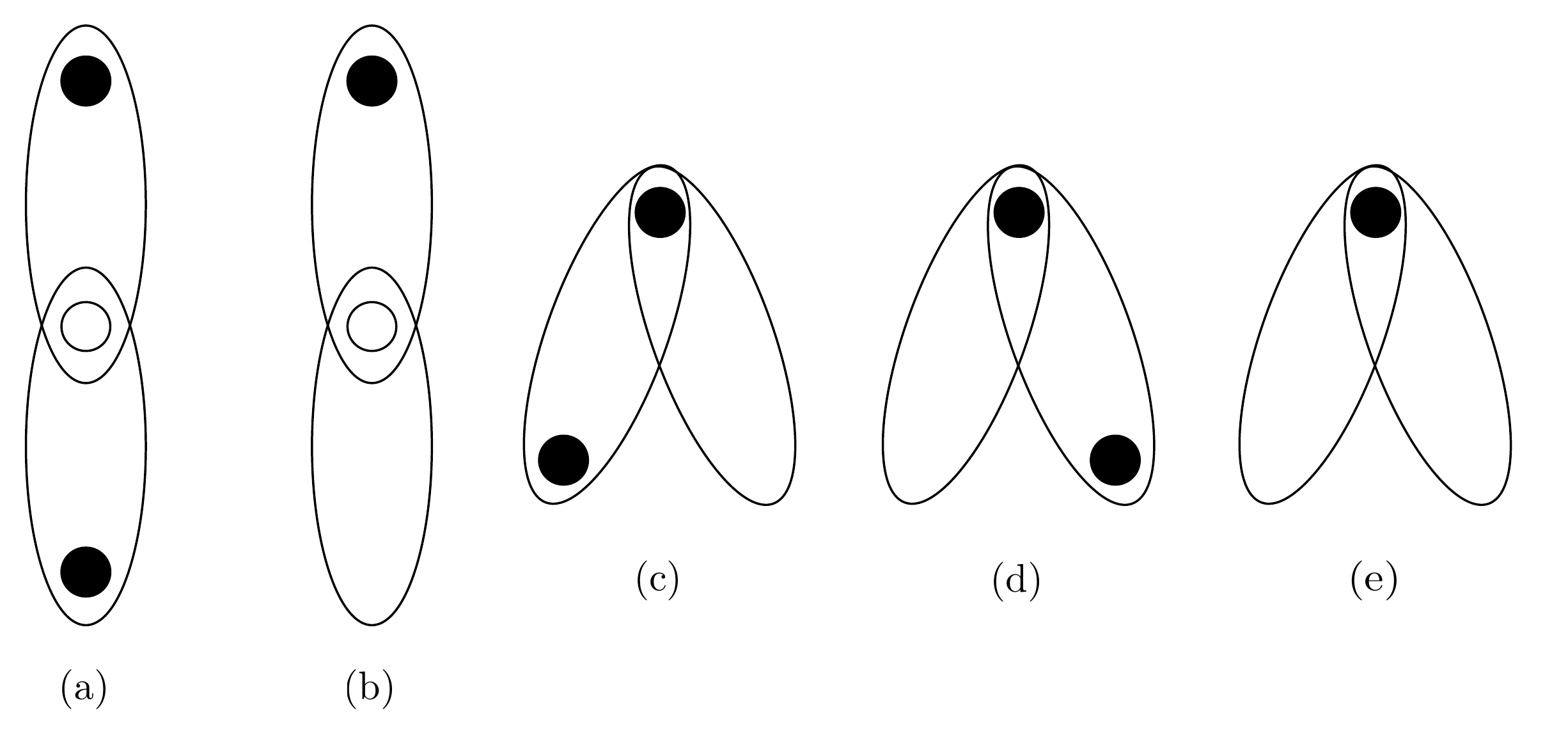}
   \caption{Five ways of merging clusters. The $\bullet$ nodes are boundary nodes that remain  boundary nodes in the merged cluster. The $\circ$ nodes are  boundary nodes that become internal (non-boundary) nodes in the merged cluster. Note that in the last four merges at least one of the  merged clusters has a top boundary node but no bottom boundary node.}
   \label{fig:merge}
\end{figure}

\subsection{Top Trees}
A \emph{top tree} $\TT$ of $T$ is a hierarchical decomposition of $T$ into  clusters. It is 
an ordered, rooted, labeled, and binary tree defined as follows.
\begin{itemize}
\item[$\bullet$] The nodes of $\TT$ correspond to clusters of $T$.
\item[$\bullet$] The root of $\TT$ corresponds to the cluster $T$ itself.
\item[$\bullet$] The leaves of $\TT$ correspond to the edges of $T$. The label of each leaf is the pair of labels of the endpoints of its corresponding edge $(u,v)$ in $T$. The two labels are ordered so that the label of the parent  appears before the label of the child. 
\item[$\bullet$] Each internal node of $\TT$ corresponds to the merged cluster of its two children. 
The label of each internal node is the type of merge it represents (out of the five merging options).  The children are ordered so  that the left child is the child cluster visited first in a preorder traversal of $T$. 
\end{itemize}

\subsection{Constructing the Top Tree}\label{sec:toptreeconstruction}

We now describe a greedy algorithm for constructing a top tree $\TT$ of $T$ that has height $O(\log \nt)$. The algorithm constructs the top tree $\TT$ bottom-up in $O(\log n_T)$ iterations starting with the edges of $T$ as the leaves of $\TT$. 
During the construction, $\TT$ is a forest, and we maintain an auxiliary rooted ordered tree $\Taux$ initialized as $\Taux := T$. The edges of $\Taux$ will correspond to the nodes of $\TT$ and to the clusters of $T$. 
The internal nodes of $\Taux$ will correspond to boundary nodes of clusters in $T$, and the leaves of $\Taux$   will correspond to a subset of the leaves of $T$.

In the beginning, these clusters represent actual edges $(v,p(v))$ of $T$. In this case, if $v$ is not a leaf in $T$ then $v$ is the bottom boundary node of the cluster and $p(v)$ is the top boundary node. If $v$ is a leaf then there is no bottom boundary node. 

In each one of the $O(\log n_T)$ iterations, a constant fraction of $\Taux$'s edges (i.e., clusters of $T$) are merged. 
Each merge is performed on two overlapping edges $(u,v)$ and $(v,w)$ of $\Taux$ using one of the five types of merges from Fig.~\ref{fig:merge}: If $v$ is the parent of $u$ and the only child of $w$ then a merge of type (a) or (b) contracts these edges in $\Taux$ into the edge $(u,w)$. If $v$ is the parent of both $u$ and $w$, and $w$ or $u$ are  leaves, then a merge of type (c), (d), or (e) replaces these edges in $\Taux$ with either the edge $(u,v)$ or $(v,w)$. In all cases, we create a new node in $\TT$ whose two children are the clusters corresponding to $(u,v)$ and to $(v,w)$. 

We prove below that a single iteration shrinks the tree $\Taux$ (and the number of roots in $\TT$) by a constant factor. The process ends when $\Taux$ is a single edge. Each iteration is performed as follows:

\paragraph{Step 1: Horizontal Merges.} For each node $v\in \Taux$ with $k\ge 2$ children $v_1, \ldots, v_k$, for $i=1$ to $\lfloor k/2 \rfloor$,
merge the edges $(v,v_{2i-1})$ and $(v, v_{2i})$  if $v_{2i-1}$ or $v_{2i}$ is a leaf.
If $k$ is odd  and $v_k$ is a leaf and both $v_{k-2}$ and $v_{k-1}$ are non-leaves then also merge $(v,v_{k-1})$ and $(v, v_{k})$.

\paragraph{Step 2: Vertical Merges.} For each maximal path $v_1, \ldots, v_p$ of nodes in $\Taux$ such that $v_{i+1}$ is the parent of $v_{i}$ and $v_2, \ldots, v_{p-1}$ have a single child: If $p$ is even merge the following pairs of  edges $\{(v_1,v_2),(v_2, v_3)\}, \{(v_3,v_4),(v_4, v_5)\},\ldots , 
(v_{p-2}, v_{p-1})\}$. If $p$ is odd  merge the following pairs of  edges 
$\{(v_1,v_2),(v_2, v_3)\}, \{(v_3,v_4),(v_4, v_5)\},\ldots ,
(v_{p-3}, v_{p-2})\}$, and 
if $(v_{p-1}, v_p)$ was not 
merged in Step 1 then also merge $\{(v_{p-2},v_{p-1}),(v_{p-1}, v_p)\}$.

\begin{lemma}\label{lem:contractionfactor}
A single iteration shrinks $\Taux$ by a factor of $c\ge 8/7$.
\end{lemma}
\begin{proof}
Suppose that  in the beginning of the iteration the tree $\Taux$ has $n$ nodes. 
Any tree with $n$ nodes has at least $(n+1)/2$ nodes with less than $2$ children. Consider the edges $(v_i,p(v_i))$ of $\Taux$ where $v_i$ has one or no children. We show that at least half of these  $n/2$ edges are merged in this iteration. This will imply that $n/4$ edges of $\Taux$ are replaced with $n/8$ edges and so the size of $\Taux$ shrinks to $7n/8$. To prove it, we charge each edge $(v_i,p(v_i))$ that is not merged to a unique edge $f(v_i,p(v_i))$ that is merged.

\paragraph{Case 1.} 
Suppose that $v_i$ has no children (i.e., is a leaf).   If $v_i$ has at least one sibling and $(v_i,p(v_i))$ is not merged it is because $v_i$ has no right sibling and its left sibling $v_{i-1}$ has already been merged (i.e., we have just merged $(v_{i-2},p(v_{i-2}))$ and $(v_{i-1},p(v_{i-1}))$ in Step 1 where $p(v_i)=p(v_{i-1})=p(v_{i-2})$). 
We also know that at least one of $v_{i-1}$ and $v_{i-2}$ must be a leaf. We set $f(v_i,p(v_i)) = (v_{i-1},p(v_{i-1}))$ if $v_{i-1}$ is a leaf, otherwise we set $f(v_i,p(v_i)) = (v_{i-2},p(v_{i-2}))$. 

\paragraph{Case 2.} 
Suppose that $v_i$ has no children (i.e., is a leaf) and no siblings (i.e., $p(v_i)$ has only one child). The only reason for not merging $(v_i,p(v_i))$ with $(p(v_i),p(p(v_i)))$ in Step 2 is because  $(p(v_i),p(p(v_i)))$ was just merged in Step 1. In this case, we set $f(v_i,p(v_i)) = (p(v_i),p(p(v_i)))$. Notice that we haven't already charged $(p(v_i),p(p(v_i))$ in {\em Case 1} because $p(v_i)$ is not a leaf.

\paragraph{Case 3.} 
Suppose that $v_i$ has exactly one child $c(v_i)$ and that $(v_i,p(v_i))$ was not merged in Step 1. 
The only reason for not merging $(v_i,p(v_i))$ with $(c(v_i),v_i)$ in Step 2 is if $c(v_i)$ has only one child $c(c(v_i))$ and we just merged  $(c(v_i),v_i)$ with  $(c(c(v_i)),c(v_i))$. In this case, we set $f(v_i,p(v_i)) = (c(v_i),v_i)$.
Notice that we haven't already charged $(c(v_i),v_i)$ in {\em Case 1} because $c(v_i)$ is not a leaf. We also haven't  charged $(c(v_i),v_i)$ in {\em Case 2} because $v_i$ has only one child.  \qed
\end{proof}

\noindent Since each iteration can be done in linear time and shrinks $\Taux$ by a factor $> 1$ we obtain the following.
\begin{corollary}
Given a tree $T$, the greedy top tree construction creates a top tree of size  $O(\nt)$ and height $O(\log n_T)$ in $O(n_T)$ time.
\end{corollary}

The next lemma follows from the construction of the top tree and Lemma~\ref{lem:contractionfactor}.
\begin{lemma}\label{lem:subtresize}
For any node $c$ in the top tree corresponding to a cluster $C$ of $T$, the number of nodes in the subtree $\TT(c)$ is $O(|C|)$.
\end{lemma}

\subsection{Top Dags}
The \emph{top DAG} of $T$, denoted $\TD$, is the minimal DAG representation of the top tree $\TT$.  It can be computed in $O(\ntt)$ time from $\TT$ using the algorithm of~\cite{DST1980}.
The entire  top DAG construction can thus be done in $O(\nt)$ time.


\section{Compression Analysis}

\subsection{Worst-case Bounds for Top Dag Compression}
We now prove Theorem~\ref{thm:worstcasebounds}. 
Let $T$ be an ordered tree with $\nt$ nodes labeled from an alphabet of size $\sigma$, let $\TT$ be its top tree and $\TD$ be its top DAG. We call two rooted subtrees of $\TT$ {\em identical} if they have the same structure and labels, otherwise they are called {\em distinct}.  
To show that  the size of $\TD$ is at most 
$ O(\nt/\log_\sigma^{0.19}\!\nt)$ is suffices to show that $\TT$ has only $O(n_T/\log_\sigma^{0.19}n_T)$ distinct rooted subtrees.

Recall that each node in the top tree $\TT$ corresponds to a cluster in $T$. A leaf of $\TT$ corresponds to a cluster of a single edge of $T$ and is labeled by this edges endpoints (so there are $O(\sigma^2)$ possible labels). An internal node is labeled by the type of merge that formed it (there are five merging options so there are five possible labels). 

The bottom-up construction of  $\TT$ starts with the leaves of $\TT$. By Lemma~\ref{lem:contractionfactor} each level in the top tree reduces the number of clusters by a factor $c=8/7$, while at most doubling the size of the current clusters (the  size of a cluster is the number of nodes in the corresponding tree pattern). After round $i$ we are therefore left with at most $O(n_T/c^i)$ clusters, each of size at most $2^{i} + 1$. 

To bound the total number of distinct rooted subtrees, we partition the clusters into \emph{small clusters} and \emph{large clusters}. 
The small clusters are those created in rounds $1$ to $j= \log_2 (0.5 \log_{4\sigma^2}(n_T)) = O( \log_2 \log_{\sigma}n_T)$ and the large clusters are those created in the remaining rounds from $j+1$ to $h$.  The total number of large clusters is at most 
$$
\sum_{i = j+1}^h O({n_T}/{c^i}) = O({n_T}/{c^{j+1}}) = O({n_T}/{
\log_\sigma^{0.19} n_T}).
$$ 
In particular, there are at most $O(n_T/\log_\sigma^{0.19} n_T)$ nodes of $\TT$ that correspond to large clusters. So clearly there are at most $O(n_T/\log_\sigma^{0.19} n_T)$ distinct subtrees rooted at these nodes. 

Next, we bound the total number of distinct subtrees of $\TT$ rooted at nodes  corresponding to small clusters.  Each such subtree is of size at most  
most $2^{j} + 1$ and is a binary tree whose nodes have labels from an alphabet of size at most $\sigma^2 + 5$. 
The total number of distinct labeled binary trees of size at most $x$  is given by 
$$
\sum_{i=1}^x (\sigma^2+5)^i \cdot C_{i-1} = \sum_{i=1}^x O((\sigma^2+5)^i \cdot4^i ) = O\left((4\sigma^2)^{x+1}\right), 
$$
where $C_{i}$ denotes the $i$th Catalan number. Since $x = 2^{j} + 1$, this number is bounded by 
$O((4\sigma^2)^{2^j+2}) = O(\sigma^4 \sqrt{n_T}) = O({n^{3/4}_T})$. In the last equality we assumed that $\sigma < n^{1/16}_T$. If $\sigma > n^{1/16}_T$ then the lemma trivially holds because $O(n_T/(\log_\sigma^{0.19}n_T)) = O(n_T)$.
We get that the total number of distinct subtrees of $\TT$  rooted at small clusters is therefore  also $O(n_T/\log_\sigma^{0.19} n_T)$.
This completes the proof of Theorem~\ref{thm:worstcasebounds}. 

\subsection{Comparison to Subtree Sharing}\label{sec:comparisonsubtreesharing}
We now prove Theorem~\ref{thm:comparison}. To do so we first show two useful properties of top trees and top DAGs. 

Let $T$ be a tree with top tree $\TT$. For any internal node $z$ in $T$, we say that the subtree $T(z)$ is \emph{represented} by a set of clusters $\{C_1, \ldots, C_\ell\}$ from $\TT$ if $T(z) = C_1 \cup \cdots \cup C_\ell$. Here $G = X_1 \cup \cdots \cup X_\ell$ denotes the graph with node set $V(G) = \cup_{i=1,\ldots, k}V(X_i)$ and  edge set $E(G) = \cup_{i=1,\ldots, k}E(X_i)$. Since each edge in $T$ is a cluster in $\TT$ we can always trivially represent $T(z)$ by at most $|T(z)| - 1$ clusters. We prove that there always exists a set of clusters, denoted $S_z$, of size $O(\log \nt)$ that represents $T(z)$. 

Let $z$ be any internal node in $T$ and let $z_1$ be its  leftmost child. Since $z$ is internal we have that $z$ is the top boundary node of the leaf cluster $L_z=(z,z_1)$ in $\TT$. Let $U$ be the smallest cluster in $\TT$ containing all nodes of $T(z)$. We have that $L_z$ is a descendant leaf of $U$ in $\TT$. Consider the path $P_z$ 
in $\TT$ from $U$ to $L_z$. An \emph{off-path}  cluster of $P_z$ is a cluster $C$ that is not on $P_z$, but whose parent cluster is on $P_z$. We define 
$$
S_z = \{C \mid \text{$C$ is off-path  cluster of $P_z$ and the tree pattern $C$ is a subtree of $T(z)$}\} \cup \{L_z\} \;.
$$

Since the length of $P_z$ is $O(\log n_T)$ the number of clusters in $S_z$ is $O(\log n_T)$. We want to prove that $\cup_{C \in S_z} C = T(z)$. 
By definition of $S_z$ we have that all nodes in $\cup_{C \in S_z} C$ are in $T(z)$. For the other direction, we first prove the following lemma. Let $E(C)$ denote the set of edges in $T$ of a cluster $C$.
\begin{lemma}\label{lem:offpath}
Let $C$ be an off-path  cluster of $P_z$. Then either $E(C) \subseteq E(T(z))$ or $E(C) \cap E(T(z)) = \emptyset$.
\end{lemma}
\begin{proof}
We will show that any cluster in $\TT$ containing edges from both $T(z)$ and $T\setminus T(z)$ contains both $(p(z),z)$ and $(z,z_1)$, where $z_1$ is the leftmost child of $z$ and $p(z)$ is the parent of $z$. Let $C$ be a cluster containing edges from both $T(z)$ and $T\setminus T(z)$. 
Consider the subtree $\TT(C)$ and let $C'$ be the smallest cluster in  $\TT(C)$ containing edges from both $T(z)$ and $T\setminus T(z)$. That is,  $C'$ is the cluster found by descending down from $C$ towards a child with both types of edges as long as such a child exists. Then $C'$ must be a merge of type (a) or (b), where the higher cluster $A$ only contains edges from $T\setminus T(z)$ and the bottom cluster, $B$, only contains edges from $T(z)$. Also, $z$ is the top boundary node of $B$ and the bottom boundary node of $A$. Clearly, $A$ contains the edge $(p(z),z)$, since all clusters are connected tree patterns. A merge of type (a) or (b) is only possible when $B$ contains all children of its top boundary node. Thus $B$ contains the edge $(z,z_1)$. It follows that $C'$ (and therefore $C$ since it is an ancestor of $C'$) contains both $(p(z),z)$ and $(z,z_1)$.

We have $L_z=(z,z_1)$ and therefore all clusters in $\TT$ containing $(z,z_1)$ lie on the path from $L_z$ to the root. The path $P_z$ is a subpath of this path, and thus no off-path clusters of $P$ can contain $(z,z_1)$. Therefore no off-path clusters of $P$ can contain edges from both $T(z)$ and $T\setminus T(z)$.\qed
\end{proof}

Any edge from $T(z)$ (except $(z,z_1)$) contained in a cluster on $P$ must be contained in an off-path cluster of $P$. Lemma~\ref{lem:offpath} therefore implies that $T(z) = \cup_{C \in S_z} C$ and the following corollary.

\begin{corollary}
Let $T$ be a tree with top tree $\TT$. For any internal node $z$ in $T$, the subtree $T(z)$ can be represented by a set of $O(\log \nt)$ clusters in $\TT$.
\end{corollary}

Next we prove that our bottom-up top tree construction guarantees that two identical subtrees $T(z), T(z')$ are represented by two {\em identical} sets of clusters $S_z, S_{z'}$. Two sets of clusters are identical (denoted $S_z = S_{z'}$) if there is a 1-1 correspondence between the clusters in $S_z$ and $S_{z'}$, such that two clusters mapped to each other are identical tree patterns in $T$ (have the same structure and labels).


\begin{lemma}\label{lem:top-tree-ident}
Let $T$ be a tree with top tree $\TT$. Let $T(z)$ and $T(z')$ be identical subtrees in $T$. 
Then, $S_z = S_{z'}$. 
\end{lemma}
\begin{proof} 

Consider the tree $\Taux$ at some iteration of the construction of the top tree. We will say that an edge $e$ in $ \Taux$ \emph{belongs to} $T(z)$ (resp. $T(z')$) if the cluster corresponding to $e$ only contains edges from $T(z)$ (resp. $T(z')$) in the original tree. Let $L_z$ be the cluster in $\Taux$  containing  the edge $L=(z,z_1)$, where $z_1$ is the leftmost child of $z$. Define $L_{z'}$ similarly.

We will say that a cluster $C\neq L_z$ is added to $S_z$ in the iteration where its parent on $P_z$ is created, and we say that $L_z$ is added to $S_z$ right before the first round. Similarly for clusters in $S_z$.

We will show that new clusters only are added to $S_z$ (resp. $S_{z'}$)  if $L_z$ (resp. $L_{z'}$) is merged with an edge belonging to $T(z)$ (resp. $T(z')$), and that these merges are identical for the two subtrees in each iteration. 

Recall that $U$ is the smallest cluster in $\TT$ containing all nodes of $T(z)$ and that  $P$ is the path of clusters in $\TT$ from $U$ to $L$. By definition, all clusters on the path $P$ contain $L$. This implies that new off-path clusters are only  constructed when $L_z$ (resp.\ $L_{z'}$) is merged. Merges of identical edges belonging to $T(z)$ and $T(z')$ are the same in the two subtrees of $\Taux$, since we merge first horizontally, and then vertically bottom-up.  
By the same argument if $L_z$ is merged with an edge belonging to $T(z)$ then $L_{z'}$ is merged with the corresponding edge from $T(z')$.
For a merge with an edge belonging  to $T(z)$  (resp. $T(z')$) and an edge not belonging to $T(z)$ (resp. $T(z')$), one of the edges must be $ L_z$ (resp. $L_{z'}$). If $L_z$ is merged in this iteration, but $L_{z'}$ is not, then $L_z$ is merged with an edge not belonging to $T(z)$ (and vice versa). 
Thus, after the iteration all edges belonging to $T(z)$ in  $\Taux$ are identical to the edges belonging to $T(z')$ in $\Taux$.  

New off-path clusters are only constructed when $L_z$ (resp. $L_{z'}$) are merged. It only adds new clusters to $S_z$ (resp. $S_{z'}$) if it is a merge with an edge belonging to $T(z)$ (resp. $T(z')$). Since these merges are identical for the two subtrees in each iteration, and $L_z$ is merged with an edge belonging to $T(z)$ iff  $L_{z'}$ is merged with the corresponding edge belonging to $T(z')$, we have $S_z = S_{z'}$.\qed
\end{proof}

\begin{theorem}\label{theo:ub}
For any tree $T$, $\ntd = O(\log \nt) \cdot \nd$.
\end{theorem}
\begin{proof}
Denote an edge in the DAG as shared if it is in a shared subtree of $T$. We denote the edges in the DAG $\D$ that are shared as \emph{red} edges, and the edges that are not shared as \emph{blue}. Let $r_D$ and $b_D$ be the number of red  and blue edges in the DAG $\D$, respectively.

A cluster in the top tree $\TT$ is \emph{red} if it only contains red edges from $\D$, \emph{blue} if it only contains blue edges from $\D$, and \emph{purple} if it contains both. Since clusters are connected subtrees we have the property that if cluster $C$ is red (resp.\ blue), then all clusters in the subtree $\TT(C)$ are red (resp.\ blue). 
Let $r$, $b$, and $p$ be the number of red, blue, and purple clusters in the top DAG $\TD$, respectively.


First we bound the number of red clusters in the top DAG $\TD$. Consider a shared subtree $T(z)$ from the DAG compression. $T(z)$ is represented by at most $O(\log \nt)$ clusters in $\TT$, and all these contain only edges from $T(z)$. Thus all the clusters in $S_z$ are red.
It follows from Lemma~\ref{lem:top-tree-ident} that all the clusters representing $T(z)$ (and their subtrees in $\TT$) are identical for all copies of  $T(z)$. Therefore each of these will appear only once in the top DAG $\TD$. 

The clusters representing $T(z)$ are edge-disjoint connected subtrees of $T(z)$. 
It follows from Lemma~\ref{lem:subtresize} that $|\TT(C)| = O(|C|)$ for each cluster in $S_z$. 
Therefore the total size of the subtrees of the clusters representing $T(z)$ in $\TT$ is $O(|T(z)|)$.  As argued above these are only represented once in the top DAG $\TD$. Thus the number of red clusters $r = O(r_D)$.

To bound the number of blue clusters in the top DAG, we first note that the blue clusters form rooted subtrees in the top tree.  Let $C$ be the root of such a blue subtree in $\TT$. Then $C$ is a connected component of blue edges in $T$. It follows from Lemma~\ref{lem:subtresize} that 
$|\TT(C)| = O(|C|)$. Thus the number of blue clusters $b = O(b_D)$.

It remains to bound the number $p$ of purple clusters (clusters containing both shared and non shared edges). The number of purple clusters in the top DAG $\TD$ is bounded by the number of purple clusters in the top tree $\TT$. For any purple cluster we have that all its ancestors in $\TT$ are also purple. Consider the set $P$ of purple clusters in $\TT$ that have no purple descendants. Each of the clusters in $P$ have a blue leaf cluster in its subtree. These blue leaf clusters are all distinct, and since the corresponding edges are not shared in the DAG $\D$, we have $|P| \leq b_D$. Each cluster in $P$ is the endpoint of a purple path from the root (and the union of these paths contains all purple clusters in \TT). Since the height of $\TT$ is $O(\log \nt)$ the number of nodes on each path is at most $O(\log n_T)$. It follows that the number of purple clusters in $\TT$ (and thus also in $\TD$) is at most $|P|\cdot O(\log \nt) = O(b_D \log \nt)$.

The number of edges in the $\TD$ is thus 
$b+ r + p  = O(b_D + r_D  + b_D \log \nt) = O(\nd \log \nt)$. \qed
\end{proof}

\begin{lemma}\label{lem:proofbypicture}
There exist trees $T$, such that $\nd = \Omega (\nt/ \log \nt) \cdot \ntd$.
\end{lemma}
\begin{proof}
Caterpillars and paths (where all nodes  have  identical labels) have 
$\ntd = O(\log \nt)$, whereas $\nd = \nt$ (see Figure~\ref{fig:compare}). \qed
\end{proof}
\begin{figure}[tb]
   \centering
   \includegraphics[scale=0.38]{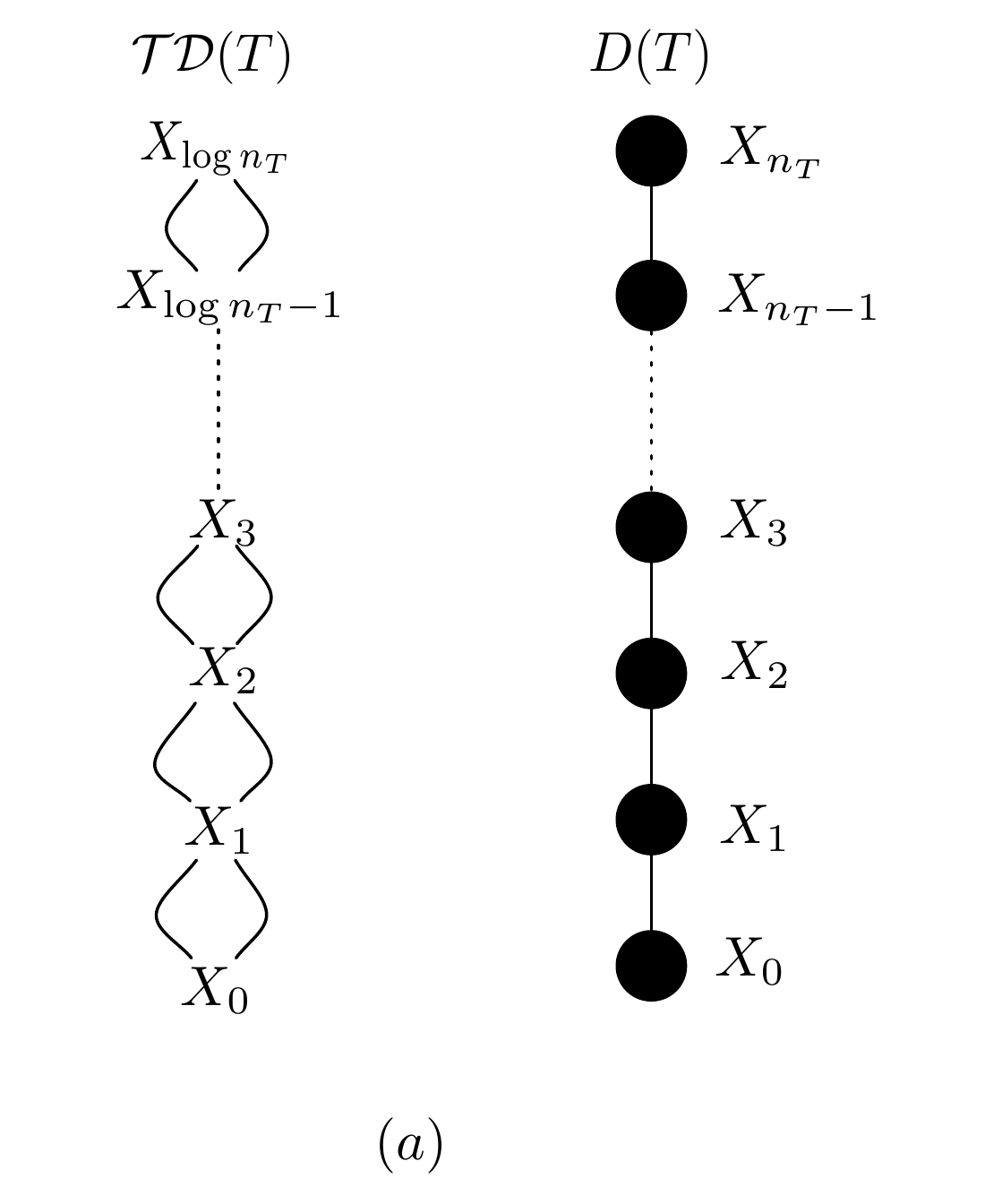}\hspace{1cm}
   \includegraphics[scale=0.38]{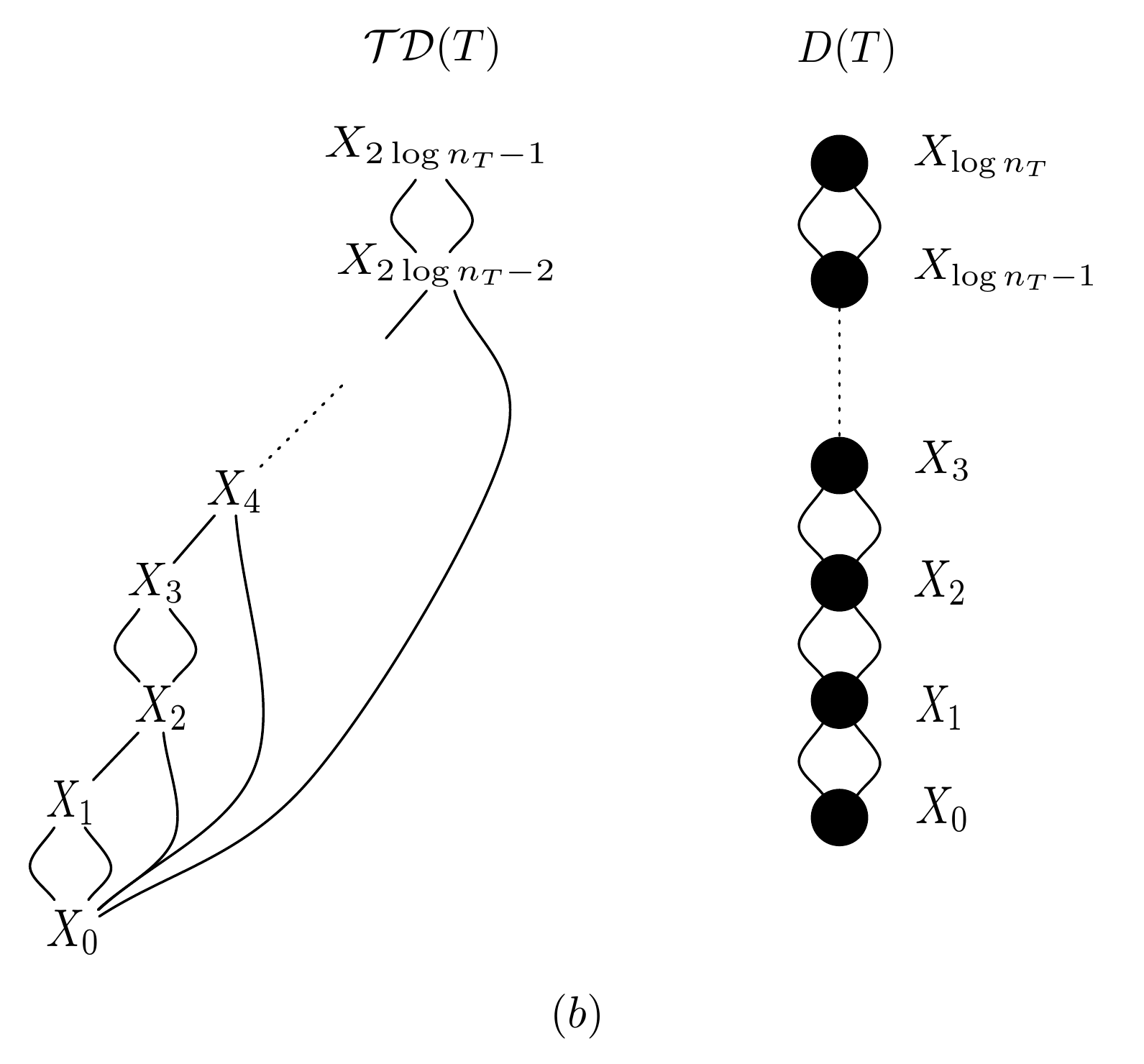}
   \caption{A  Top DAG $\TD$ and a DAG D($T$) of 
   (a) a path  and (b) a complete binary tree. All labels are identical. On a path (and also a caterpillar and a star) the size of $\TD$ is $O(\log \nt)$ whereas the size of D($T$) is $O(\nt)$. 
   On a complete binary tree (b) both $\TD$ and D($T$) are of size $O(\log \nt)$.}

  \label{fig:compare}
\end{figure}

\section{Supporting Navigational Queries}\label{section4}
In this section we prove Theorem~\ref{thm:navigation}. Let $T$ be a tree with top DAG $\TD$. To uniquely identify nodes of $T$ we refer to them by their preorder numbers. For a node of $T$ with preorder number $x$ we want to support the following queries. 

\begin{description}
\item[$\Access(x)$:] Return the label associated with node $x$.
\item[$\Decompress(x)$:] Return the tree $T(x)$.
\item[$\Parent(x)$:] Return the parent of node $x$.
\item[$\Depth(x)$:] Return the depth of node $x$.
\item[$\Height(x)$:] Return the height of node $x$.
\item[$\Size(x)$:] Return the number of nodes in $T(x)$.
\item[$\FirstChild(x)$:] Return the first child of $x$.
\item[$\NextSibling(x)$:] Return the sibling immediately to the right of $x$.
\item[$\LevelAncestor(x,i)$:] Return the ancestor of $x$ whose distance from $x$ is $i$.
\item[$\NCA(x,y)$:] Return the nearest common ancestor of the nodes $x$ and $y$.
\end{description}

\subsection{The Data Structure}
In order to enable the above queries, we augment the top DAG $\TD$ of $T$ with some additional information. Consider a cluster $C$ in $\TD$. Recall that if $C$ is a leaf in $\TD$ then $C$ is a single edge in $T$ and $C$ stores the labels of this edge's endpoints. Otherwise, $C$ is a cluster of $T$ obtained by merging two clusters: the cluster $A$  corresponding to $C$'s left child and the cluster $B$ corresponding to $C$'s right child. Consider a preorder traversal of $C$. Let $\ell(B)$ denote the first node visited in this traversal that is also a node in $B$. Let $r(B)$ (resp. $r(A)$) denote the last node visited that is also a node in $B$ (resp. in $A$). 
We augment each cluster $C$ with:

\begin{itemize}
\item[$\bullet$] The integers $r(A)$, $\ell(B)$, and $r(B)$. 
\item[$\bullet$] The type of merge that was applied to $A$ and $B$ to form $C$. If $C$ is a leaf cluster then the labels of its corresponding edge's endpoints in $T$.
\item[$\bullet$] The height and size of $C$ (i.e., of the tree pattern $C$ in $T$).
\item[$\bullet$] The distance from the top boundary node of $C$ to the top boundary nodes of $A$ and $B$.    
\end{itemize} 
Since we 
use constant space for each cluster of $\TD$, the total space remains~$O(\ntd)$. 

\paragraph{Local preorder numbers}
All of our queries are based on traversals of the augmented top DAG $\TD$. During the traversal we identify nodes by computing preorder numbers  local to the cluster that we are currently visiting. Specifically, let $u$ be a node in the cluster $C$. Define the \emph{local preorder number of $u$}, denoted $u_C$, to be the position of $u$ in a preorder traversal of $C$. The following lemma states that in $O(1)$ time we can compute $u_A$ and $u_B$ from $u_C$ and vise versa.  

\begin{lemma}\label{lem:localpreorder}
Let $c$ be an internal node of $\TD$ that corresponds to the cluster $C$ of $T$ obtained by merging the cluster  $A$ (corresponding to $c$'s left child) and the cluster  $B$ (corresponding to $c$'s right child).
For any node $u$ in $C$, given $u_C$ we can tell in constant time if $u$ is in $A$ (and obtain $u_A$) in $B$ (and obtain $u_B$) or in both. Similarly, if $u$ is in $A$ or in $B$ we can obtain $u_C$ in constant time from $u_A$ or $u_B$.
\end{lemma}

\begin{proof}
If $C$ is a merge of $A$ and $B$ of type (a) or (b) then
\begin{itemize}
\item $u_C = 1$ iff $u$ is the top boundary node of $A$ and $C$ and $u_A = 1$.
\item $u_C \in [2, \ell(B) - 1]$ iff $u$ is an internal node of $A$ and $u_A < l(B)$.
In this case $u_A = u_C$.
\item $u_C = \ell(B)$ iff $u$ is the shared boundary node of $A$ and $B$,  $u_A = \ell(B)$, and $u_B = 1$.
\item $u_C \in [\ell(B) + 1, r(B)]$ iff $u$ is an internal node in $B$. In this case $u_B = u_C - \ell(B) + 1$.
\item $u_C \in [r(B) + 1, r(A)]$ iff $u$ is an internal node in $A$ and $u_A > l(B)$. In this case $u_A = u_C - r(B) + \ell(B)$.
\end{itemize}

\noindent Otherwise, if $C$  is a merge of $A$ and $B$ of type (c), (d), or (e) then
\begin{itemize}
\item $u_C = 1$ iff $u$ is the shared  boundary node of $A$, $B$, and $C$ and $u_A = u_B = 1$.
\item $u_C \in [2, r(A)]$ iff $u$ is an internal node in $A$. In this case $u_A = u_C$.
\item $u_C \in [r(A) + 1, r(B)]$ iff $u$ is an internal node in $B$. In this case $u_B = u_C - r(A) + 1$.
\end{itemize}
\end{proof}

\subsection{Implementation of the procedures} We now show how to implement the queries using local preorder numbers in top-down and bottom-up traversals of $\TD$. 


\subsubsection{Access and Depth}
The queries $\Access(x)$ and $\Depth(x)$ ask for the label and depth of the node whose preorder number in $T$ is $x$. They are both  performed by a single top-down search of $\TD$ starting from its root and ending with the leaf cluster containing $x$. Since the depth of $\TD$ is $O(\log \nt)$ the total time is $O(\log \nt)$.

\paragraph{Access.}
At each cluster $C$ on the top-down search we compute the local preorder number $x_C$. Initially, the root cluster corresponds to the entire $T$ so we set $x_T = x$. Let $C$ be a cluster on the way. If $C$ is a leaf cluster we return the label of the top boundary node if $x_C = 1$ and the label of the single internal node if $x_C = 2$. 
 If on the other hand $C$ is an internal cluster with child clusters $A$ and $B$, we continue the search in the child cluster containing $x_C$. We compute the new local preorder number according to Lemma~\ref{lem:localpreorder}. If  $x_C$ is the shared boundary node between $A$ and $B$ we continue the search in either $A$ or $B$.

\paragraph{Depth}  The only difference between $\Depth(x)$ and  $\Access(x)$ is that during the top-down search we also sum the distances between the top boundary nodes of the visited clusters. Let $d$ be this distance. At the leaf cluster at the end of the search we return $d$ if $x_C = 1$ and $d+1$ if $x_C = 2$. 
Since the distances are stored the total time remains $O(\log \nt)$.

\subsubsection{Firstchild, Level Ancestor, Parent, 
and NCA}
We answer these queries by a top-down search to find the local preorder number in a relevant cluster $C$, and then a bottom-up search to compute the corresponding preorder number in $T$.

\paragraph{Firstchild} We compute $\FirstChild(x)$ in two steps. 
\paragraph{Step 1: Top-down Search.} We do a top-down search to find the first cluster with top boundary node $x$. We use local preorder numbers as in the algorithm for $\Access$. Let $C$ be a cluster in the search. If $x_C = 1$ we stop the search. Otherwise we know that $x_C > 1$. If $C$ is a leaf cluster we stop and report that $x$ does not have a first child since it is a leaf in $T$. 
 If on the other hand $C$ is an internal cluster with child clusters $A$ and $B$, we continue the search in the child cluster containing $x_C$. If $x_C$ is the shared boundary node between $A$ and $B$ we always continue the search in $B$. 
This ensures that we continue to the cluster containing the children of $x$ (recall that $B$ is the deeper cluster in merges of type (a) and (b)). Combined with the condition that we stop the search in the first cluster $C$ where $x$ is the top boundary node (and therefore the last merge before we stop must be of type (a) or (b)), this implies that all children of $x$ are in $C$. 

\paragraph{Step 2: Bottom-up Search.} Let $C$ be the cluster found in Step 1. Since all children of $x$ are in $C$, the node with local preorder number $2$ in $C$ is the first child of $x$. We do a bottom-up search from $C$ to the root cluster to compute the preorder number in $T$ of the node with $x_C=2$.  

\paragraph{Level Ancestor and Parent}
Notice that $\Parent(x)$ can be computed as $\LevelAncestor(x,1)$. Since $\LevelAncestor(x,0) = x$ we focus on   $\LevelAncestor(x,i)$ for $i\ge 1$. This is done in three steps:

\paragraph{Step 1: Compute Depth.} Compute the depth of $\LevelAncestor(x,i)$ as 
$d = \Depth(x) - i$. 

\paragraph{Step 2: Top-down Search.} We do a top-down search to find the cluster with top boundary node $y$ of depth $d$ such that $x$ is a descendant of $y$ (we will show that such a cluster exists). During the search we maintain the depth of the current top boundary node as in the algorithm for $\Depth$. At each cluster $C$ in the search we also compute a local preorder number $x_C'$ to guide the search. The idea is that $x_C'$ either corresponds to $x$ or to an ancestor of $x$ within $C$. Initially, for the root cluster $T$ we set $x_T' = x$. Let $C$ be an internal cluster in the search with top boundary node $v$ and with children $A$ and $B$. If the depth of $v$ is $d$ we stop the search. Otherwise, we proceed as follows.
\begin{enumerate}
\item If $C$ is of type (a) or (b), $x_C'$ is in $B$, and the shared boundary node of $A$ and $B$ has depth $> d$, we continue the search in $A$ and set $x_A'$ to be the bottom boundary of $A$. 
\item In all other cases, we continue the search in the child cluster containing $x_C'$, and compute the new local preorder number for $x_C'$. 
\end{enumerate}
Note that if the shared boundary node in case 1 has depth $d$ we continue the search in $B$. Combined with the assumption that $i > 0$, it inductively follows that $y$ becomes the top boundary node at some cluster during the top-down search. Hence, at some 
cluster in the top-down search the depth of the top boundary node is $d$. 

\paragraph{Step 3: Bottom-up Search.} Let $C$ be the cluster whose top boundary node $v$ has depth $d$ found in Step 2. We do a bottom-up search to compute the preorder number of $v$ in $T$. Finally, we report the result as $y$.

\paragraph{Nearest Common Ancestor.}
We compute $\NCA(x,y)$ in the following steps. We assume w.l.o.g. that $x \neq y$ in the following since $\NCA(x,x)=x$.

\paragraph{Step 1: Top-down Search} We do a top-down search to find the first cluster, whose top boundary node is $\nca(x,y)$ (this cluster always exists since $x \neq y$). At each cluster $C$ in the search we compute local preorder numbers $x_C'$ and $y_C'$. The idea is that  $x_C'$ and $y_C'$ are either $x$ or $y$ or ancestors of $x$ and $y$ and their depth is at least the depth of $\nca(x,y)$. Initially, for the root cluster $T$ we set $x_T' = x$ and $y_T' = y$. Let $C$ be a cluster visited during the search. If $C$ is a leaf cluster we stop the search. Otherwise, $C$ is an internal cluster with children $A$ and $B$. We proceed as follows.
\begin{enumerate}
\item If $x_C'$ and $y_C'$ are in the same child cluster, we continue the search in that cluster, and compute new local preorder numbers for $x_C'$ and $y_C'$.
\item If $C$ is of type (a) or (b) and $x_C'$ and $y_C'$ are in different child clusters we continue the search in $A$. We update the local preorder number of the node in $B$ to be the bottom boundary of $A$.
\item If $C$ is of type (c), (d), or (e) and $x_C'$ and $y_C'$ are in different child clusters we stop the search.
\end{enumerate}

\paragraph{Step 2: Bottom-up Search} Let $C$ be the cluster computed in step 1. We do a bottom-up search to compute the preorder number of the top boundary node of $C$ in the entire tree $T$, and return the result. 

\subsubsection{Decompress, Height, Size, and Next Sibling}
To answer these queries,  the
key idea is to compute a small set of clusters representing $T(x)$. This set will be a subset of the set $S_x$ defined in Sec.~\ref{sec:comparisonsubtreesharing} and will contain all the relevant information.

We need the following definitions. Let $u$ be a node in $T$. We say that $u$ is on the \emph{spine path} in a cluster $C$ if $u$ is the top boundary node in $C$, or $u$ is on the path from the top boundary node in $C$ to the bottom boundary node in $C$. Since clusters are connected subtrees we immediately have the following.

\begin{lemma}\label{lem:spinepath}
Let $C = A \cup B$ be a cluster with left child $A$ and right child $B$. A node $u$ in $T$ is on the spine path of $C$ iff one of the following cases are true: 
\begin{itemize}
\item $C$ is of type $(c)$ and $u$ is on the spine path in $A$.
\item $C$ is of type $(d)$ and $u$ is on the spine path in $B$.
\item $C$ is of type $(a)$  and $u$ is on the spine path in $A$ or $B$. 
\item $u$ is the top boundary node of $C$.
\end{itemize}
\end{lemma}

Let $x$ be any internal node in $T$. As in Section~\ref{sec:comparisonsubtreesharing}, let $L$ be the leftmost leaf cluster in $\TD$ such that $x$ is the top boundary node and let $P$ be the path of clusters from the smallest cluster $U$ containing all nodes of $T(x)$ to $L$. We also define $M$ to be the highest cluster on $P$ that has $x$ as the top boundary node, i.e., $M$ is the highest cluster on $P$ that only contains edges from $T(x)$. Recall that $S_x$ is the set of $O(\log n_T)$ off-path  clusters of $P$ that represent $T(x)$. We partition $S_x$ into the set $\widehat{S}_x$ that contains all clusters in $S_x$ that are descendants of $M$ and the set $\widecheck{S}_x$ that contains the remaining clusters. We characterize these sets as follows.

\begin{lemma}\label{lem:reprCluster}
Let $B$ be an off-path  cluster of $P$ with parent $C$ and sibling $A$. Then
\begin{enumerate}
\item $B$ is in $\widehat{S}_x$ iff $B$ is a descendant of $M$.
\item $B$ is in $\widecheck{S}_x$ iff $C$ is a merge of type (a) or (b), $B$ is the right child of $C$, and $x$ is on the spine path of $A$. 
\end{enumerate}
\end{lemma}
\begin{proof}
For the first property, first note that if $B$ is in $\widehat{S}_x$ it is by definition a descendant of $M$. Conversely, if $B$ is a descendant of $M$, we have that $E(B) \subseteq E(M) \subseteq E(T(x))$. By definition of $\widehat{S}_x$, we have that $B$ is in $\widehat{S}_x$. 

Next consider property 2. Suppose that $B$ is in $\widecheck{S}_x$. Then, by Lemma~\ref{lem:offpath} and the definition of $S_x$ we have that $E(B) \subseteq E(T(x))$. Furthermore, since $C$ is a proper ancestor of $M$, $C$ contains edges from both $T(x)$ and $T\setminus T(x)$, and therefore $A$ must also contain edges from both $T$ and $T\setminus T(x)$. 

Assume for contradiction that $C$ is of type (c), (d), or (e). Then, the top boundary node $v$ of $C$ is also the top boundary node in $A$ and $B$. Since $x \neq v$ by definition of $M$, we have by Lemma~\ref{lem:offpath} that $E(B) \cap E(T(x)) = \emptyset$ and thus $B$ cannot be in $\widecheck{S}_x$. 

Hence, assume that $C$ is of type (a) or (b). Assume for contradiction that $B$ is the left child of $C$. Since all clusters on $P$ contain $E(L)$ and $C$ contains edges from both $T(x)$ and $T\setminus T(x)$, we have that the top boundary node of $B$ is a proper ancestor of $x$. Hence, $B$ cannot be in $\widecheck{S}_x$. 

Finally, if $B$ is of type (a) or (b) and is the right child of $C$, then $E(B) \subseteq E(T(x))$ iff the top boundary node $v$ of $B$ is a descendant of $x$. But $v$ is a descendant of $x$ iff $x$ is on the spine path of $A$. Hence, $B$ is in $\widecheck{S}_x$ iff $x$ is on the spine path of $A$.
\end{proof}
\medskip

\noindent In the following we show how to efficiently compute $\widecheck{S}_x$ using the procedure $\FindRepr$. We then use $\FindRepr$ to implement the remaining  procedures.

\paragraph{FindRepresentatives}
Procedure $\FindRepr(x)$ computes the set $\widecheck{S}_x$ and cluster $M$ in two steps. 

\paragraph{Step 1: Top-down Search} We do a top-down search to find the cluster $M$, i.e., the highest cluster on $P$ that has $x$ as the top boundary node. If no such node exists, then $x$ is a leaf node in $T$. 

\paragraph{Step 2: Bottom-up Search} We do a bottom-up search from $M$ and add clusters according to Lemma~\ref{lem:reprCluster} as follows. Initially, set $S = \emptyset$. Let $A$ be a cluster on the path with sibling $B$ and parent $C$.
\begin{enumerate}
\item If $C$ is of type $(a)$ or $(b)$ and $A$ is the left child of $C$, add $B$ to $S$.
\item  If one of the following conditions are true, stop  the traversal:
\begin{itemize}
\item $C$ is of type $(c)$ and $A$ is the right child of $C$.
\item $C$ is of type $(d)$ and $A$ is the left child of $C$.
\item $C$ is of type $(e)$ or $(b)$.
\end{itemize}
\end{enumerate}
Note that, as long as we continue the bottom-up search and consider clusters on the path, we have that $x$ is on the spine path of these clusters. This is because we continue the bottom-up search  according to the cases of Lemma~\ref{lem:spinepath}. It follows from  Lemma~\ref{lem:reprCluster} that the clusters we add to $S$ are exactly the clusters in the set representing $T(x)$. The total time is $O(\log n_T)$.

\paragraph{Decompress} 
To compute $\decomp(x)$, we use $\FindRepr(x)$ to compute the sets of cluster $\widecheck{S}_x$ and $M$. We construct $T(x)$ from $\widecheck{S}_x$ (and $M$) and the path $P$ computed during the traversal of $\TD$. First, we decompress all clusters in $\widecheck{S}_x$ (and $M$) by unfolding their subDAG and constructing their corresponding subtree of $T$. We then combine these subtrees using the merge information stored for each cluster in $P$.   

In total we use $O(\log n_T)$ time for $\FindRepr(x)$ and computing the path $P$. The total time to decompress a cluster $\TD$ by unfolding is linear in its size. Hence, the total time used is $O(\log n_T + |T(x)|)$.  

\paragraph{Height} First we compute the set of clusters $\widecheck{S}_x$ and cluster $M$ using $\FindRepr(x)$. Define the \emph{local height} of a cluster to be the length of the path from the top boundary node to the bottom boundary node if it is an internal cluster, and the height of the cluster if it is a leaf cluster.  We compute the height of $T(x)$ as the sum of the local heights of all clusters in $\widecheck{S}_x$ plus the height of $M$. This correctly computes the height since all clusters in $\widecheck{S}_x$ are merged with their siblings by type (a) or (b). Since the  height and the distance from top boundary node to bottom boundary node for each cluster in $\TD$ is stored we use $O(\log n_T)$ time in total. 

\paragraph{Size} Similar to height. We sum the sizes of clusters in $\widecheck{S}_x$ and $M$ and subtract $|\widecheck{S}_x|$ (to exclude shared boundary nodes). This also uses $O(\log n_T)$ time.

\paragraph{Nextsibling} We compute $\NextSibling(x)$ directly from $\Size(x)$ since $\NextSibling(x) = x + \Size(x)$.

\section{Conclusion and Open Problems}
We have presented the new top tree compression scheme, and shown that it achieves close to optimal worst-case compression, can compress exponentially better than DAG compression, is never much worse than DAG compression, and supports navigational queries in logarithmic time. We conclude with some open problems.
\begin{itemize}
\item Surprisingly, top tree compression is the first compression scheme for trees that achieves any provable non-trivial compression guarantee compared to the classical DAG compression. We wonder how other tree compression schemes compare to DAG compression and if it is possible to construct a tree compression scheme that  exploits tree pattern repeats and always compresses better than a logarithmic factor of the DAG compression. 
\item Pattern matching in compressed \emph{strings} is a well-studied and well-developed area with numerous results, see e.g., the surveys \cite{Gasieniec1996,Rytter2004,Lohrey2012}. Pattern matching in compressed trees (especially within tree compression schemes that exploit tree pattern repeats) is a wide open area. 
\item We wonder if top tree compression is practical. In preliminary experiments we have compared our top DAG compression with standard DAG compression on typical XML datasets that were previously used in papers on DAG compression. The experiments match our theoretical expectations, i.e., that most trees compress better with top tree compression, and only balanced trees compress slightly better with standard DAG compression.
\end{itemize}

\section{Acknowledgments}
We would like to thank the anonymous reviewer for the important and helpful comments.

\bibliographystyle{abbrv}
\bibliography{paper}

%
%
%
%

\end{document}